\documentclass[letterpaper,11pt]{article}

\usepackage[ruled,vlined]{algorithm2e}
\usepackage{amsmath}
\usepackage{amssymb}
\usepackage{enumitem}
\usepackage{url}
\usepackage{verbatim}
\usepackage{amsthm}

\usepackage[compact]{titlesec}
\titlespacing{\section}{2pt}{*2}{*1}
\newtheorem{theorem}{Theorem}[section]
\newtheorem{lemma}[theorem]{Lemma}

\newtheorem{corollary}[theorem]{Corollary}
\newtheorem{definition}[theorem]{Definition}

\DeclareMathOperator*{\argmin}{argmin}
\DeclareMathOperator*{\argmax}{argmax}
\newcommand{\RR}{\mathbb R}

\newcommand{\ValueIteration}{\mbox{\sc Value-Iteration}}
\newcommand{\PolicyIteration}{\mbox{\sc Policy-Iteration}}
\newcommand{\StrategyIteration}{\mbox{\sc Strategy-Iteration}}

\newcommand{\ee}{{\bf e}}

\newcommand{\xx}{{\bf x}}
\newcommand{\vv}{{\bf v}}
\newcommand{\uu}{{\bf u}}
\newcommand{\cc}{{\bf c}}
\newcommand{\s}{{\bf s}}
\newcommand{\zz}{{\bf 0}}
\newcommand{\TT}{{\mathcal T}}
\newcommand{\PP}{{\mathcal P}}
\newcommand{\Pol}{\Pi}

\newcommand{\etal}{{\em et al.\ }}

\title{Strategy iteration is strongly polynomial for 2-player\\ turn-based stochastic games with a constant discount factor}

\author{\em Thomas Dueholm Hansen \thanks{Department of Computer Science,
Aarhus University, Denmark. E-mail:
{\tt \{tdh,bromille\}@cs.au.dk}. Supported by the Center for
Algorithmic Game Theory, funded by the Carlsberg Foundation. } \and \em Peter Bro Miltersen ${}^*$
\and \em Uri Zwick \thanks{School of Computer Science,
Tel Aviv University, Tel Aviv 69978, Israel. E-mail: {\tt
zwick@tau.ac.il}. Supported by grant 1306/08 of the Israeli Science Foundation.}}

\date{}

\begin{document}
\maketitle
\begin{abstract}\noindent
\baselineskip=14pt
Ye showed recently that the \emph{simplex method} with Dantzig pivoting rule, as well as Howard's \emph{policy iteration} algorithm, solve discounted
\emph{Markov decision processes} (MDPs), with a constant discount factor, in \emph{strongly polynomial} time.
More precisely, Ye showed that both algorithms terminate after at most
$O\bigl(\frac{mn}{1-\gamma}\log\bigl(\frac{n}{1-\gamma}\bigr)\bigr)$ iterations, where $n$ is the number of states, $m$ is the total number of actions in the MDP, and $0<\gamma<1$ is the discount factor.
We improve Ye's analysis in two respects. First, we improve the bound given by Ye and show that Howard's
policy iteration algorithm actually terminates after at most $O\bigl(\frac{m}{1-\gamma}\log\bigl(\frac{n}{1-\gamma}\bigr)\bigr)$ iterations.
% improving over Ye's result by a factor of~$n$.
Second, and more importantly, we show that the same bound applies to the number of iterations performed by the \emph{strategy iteration} (or \emph{strategy improvement}) algorithm, a generalization of Howard's policy iteration algorithm used for solving 2-player turn-based \emph{stochastic games} with discounted zero-sum rewards. This provides the first strongly polynomial algorithm for solving these games, resolving a long standing open problem.
\end{abstract}

\thispagestyle{empty}
\newpage
\setcounter{page}{1}

\section{Introduction}

\emph{Markov Decision Processes} (MDPs) are widely used in operations research, machine learning and related disciplines, to model long-term sequential decision making in uncertain, i.e., stochastic, environments. \emph{Stochastic Games} (SGs), a generalization of MDPs to a 2-player setting, are widely used to model long-term sequential decision making in stochastic and adversarial environments.
% Both models emerged in the 1950s.
MDPs were first introduced by Bellman \cite{Bellman57}.
SGs, which form a more general model, were introduced slightly earlier by Shapley \cite{Shapley53}. Many variants of MDPs and SGs were studied in the literature. The MDPs and SGs considered in this paper are infinite-horizon \emph{discounted} MDPs/SGs. The SGs we consider are \emph{turn-based} and we thus refer to them as \emph{2-player Turn-Based Stochastic Games} (2TBSG).

MDPs may be viewed as degenerate 2TBSGs in which one of the players has no influence on the game.
For a thorough treatment of MDPs and their numerous practical applications, see the books of Howard \cite{Howard60}, Derman \cite{Derman72}, Puterman \cite{Puterman94} and Bertsekas \cite{Bertsekas01}. For a similar treatment of SGs,
see the books of Filar and Vrieze \cite{DutchBook} and Neyman and Sorin
\cite{NatoBook}.

% MDPs have numerous of practical applications.
% SGs were first brought to the attention of the computer science community by Condon's
% \cite{Condon92,Condon93} work on \emph{simple} stochastic games.

A 2TBSGs is composed of a finite set of \emph{states} and a finite set of \emph{actions}. Each state is controlled by one of the players. % In MDPs, all states are controlled by the single player.
In each time unit, the game is in exactly one of the states. Each state has a non-empty set of actions associated with it. The player controlling the state must play one of these actions. Playing an action incurs an immediate \emph{cost}, and results in a probabilistic transition to a new state according to a probability distribution that depends on the action. The process goes on indefinitely. The first player tries to \emph{minimize} the total expected \emph{discounted} cost of the infinite sequence of actions taken, with respect to a fixed \emph{discount factor}. The second player tries to \emph{maximize} this total discounted cost. Discounting captures the fact that a cost incurred at a later stage has a smaller effect than the same cost incurred at an earlier stage. For formal definitions, see Section~\ref{S-2TBSG}.

% An alternative t

A \emph{policy} or a \emph{strategy} for a player is a possibly probabilistic rule that specifies the action to be taken in each situation, given the full history of play so far. One of the fundamental results in the theory of MDPs and 2TBSGs, is that both players have \emph{positional} optimal strategies. A positional strategy is a strategy that is both \emph{deterministic} and \emph{memoryless}. A \emph{memoryless} strategy is a strategy that depends only on the current state, and not on the full history. MDPs and 2TBSGs are \emph{solved} by finding optimal positional strategies for the players.

% A fundamental model of operations research is the finite, but infinite-horizon,
% discounted \emph{Markov Decision Process} (MDP). In an MDP, in each of
% a finite set of situations or states, a single agent must choose one action
% from a finite set of possible actions. Depending on the state and the
% chosen action,
% a probabilistic transition to another state is made. Also, the agent
% experiences a {\em reward} depending on the state and the
% action. {\em Solving}
% an MDP means finding a single action for each state (collectively
% called a {\em policy}) that maximizes
% the expected total reward experienced by the agent, when future
% rewards are discounted geometrically according to a
% discount factor. For formal definitions, see Section \ref{S-2TBSG}.
% Solving MDPs has a vast number of practical applications,
% see, e.g., the monograph of Puterman \cite{Puterman94}.

% It has long been known that
MDPs can be solved using linear programming (d'Epenoux \cite{dEpenoux63}, Derman \cite{Derman72}).
The preferred way of solving MDPs in practice, however, is Howard's \cite{Howard60} \emph{Policy
Iteration} algorithm. The policy iteration algorithm maintains and
iteratively improves a policy by performing ``obvious'' improving switches
(for details, see Section \ref{sec-strat}). Howard's algorithm may be viewed as a parallel version of the simplex algorithm in which several pivoting steps are performed simultaneously. The problem of determining the worst case complexity of Howard's algorithm was stated explicitly at least 25 years ago. (It is mentioned, among other places, in
Schmitz \cite{Schmitz85}, Littman \etal \cite{LiDeKa95} and Mansour and Singh \cite{MaSi99}.)
Meister and Holzbaur
\cite{Meister86} established, decades ago, that the number of iterations performed by Howard's algorithm, when the discount factor is fixed, is polynomially bounded in the \emph{bit size} of the input. Their bound, however, is not polynomial in the number of states and actions of the MDP. The first \emph{strongly polynomial} time algorithm for solving MDPs was an interior point algorithm of Ye \cite{Ye05}.

Very recently, Ye \cite{Ye10} presented
a surprisingly simple proof that Howard's algorithm terminates
after at most
$O\bigl(\frac{mn}{1-\gamma}\log\bigl(\frac{n}{1-\gamma}\bigr)\bigr)$
iterations, where $n$ is the number of states, $m$ is the total number
of actions, and $0<\gamma<1$ is the discount factor.
In particular, when the discount factor is constant, the number of
iterations
is $O(m n \log n)$. Since each iteration only involves solving a system of linear equations, Ye's result established for the first time that Howard's
algorithm is a \emph{strongly polynomial} time algorithm, when the
discount factor is constant. Ye's proof is based on a careful
analysis of an LP formulation of the MDP problem,
with LP duality and complementary slackness playing crucial roles.

We significantly improve and extend Ye's \cite{Ye10} analysis.
We show that Howard's algorithm actually terminates after at most
$O\bigl(\frac{m}{1-\gamma}\log\bigl(\frac{n}{1-\gamma}\bigr)\bigr)$
iterations, improving Ye's bound by a factor of $n$.
Interestingly, the only added ingredient needed to obtain this significant improvement
is a well-known relationship between Howard's policy iteration algorithm and Bellman's \cite{Bellman57}
\emph{value iteration} algorithm, an algorithm for \emph{approximating} the values of MDPs.

More significantly, and more surprisingly, we are able to obtain the same $O\bigl(\frac{m}{1-\gamma}\log\bigl(\frac{n}{1-\gamma}\bigr)\bigr)$ bound also for the
\emph{Strategy Iteration} (or \emph{Strategy Improvement}) algorithm for the solution of 2TBSGs. This supplies the \emph{first} strongly polynomial algorithm for solving 2TBSGs, with a fixed discount factor, solving a long standing open problem.
% The previously best known algorithm for solving discounted 2TBSG was a weakly polynomial time algorithm due to Littman \cite{litt}.

The strategy iteration algorithm is a natural generalization of Howard's policy iteration algorithm that can be used to solve 2TBSGs. The strategy iteration algorithm for discounted 2-player games is apparently first described by Rao {\em et al.} \cite{RCN}. Hoffman and Karp \cite{HoKa66} earlier described a related algorithm for a somewhat different class of SGs.

Prior to our strongly polynomial bound for the strategy iteration algorithm, the best time
available on the problem of solving discounted 2TBSGs was a polynomial, but not strongly polynomial, bound of Littman \cite{litt}, obtained essentially using value iteration. The best time bound expressed solely in terms of the number states and actions was a \emph{subexponential} bound of Ludwig \cite{Ludwig95}. (See also {Bj{\"o}rklund} and Vorobyov \cite{BjVo05,BjVo07} and Halman \cite{Halman07}.)
Interestingly, these subexponential bounds are obtained using randomized variants of the strategy iteration algorithm that mimic the combinatorial subexponential algorithms of Kalai \cite{Kalai92,Kalai97} and Matou{\v{s}}ek, Sharir and Welzl \cite{MaShWe96}  for solving \emph{LP-type} problems.

What makes our analysis of the strategy iteration algorithm surprising is the fact that Ye's analysis relies heavily on the LP formulation of MDPs. In contrast, no succinct LP formulation is known for 2TBSGs.
(Natural attempts fail. See Condon \cite{Condon93}.)
Our proof is based on finding natural game-theoretic quantities that
correspond to the LP-based quantities used by Ye, and by
reestablishing, via direct means, (improved versions) of the bounds obtained by Ye using LP duality.

% The \emph{limiting average cost} criteria is an alternative to the

% In contrast to Ye's and our results, that assume a fixed discount factor,

% Friedmann \cite{Friedmann09}, in a breakthrough result, has recently shown that the strategy iteration algorithm may require an \emph{exponential} number of iterations on \emph{non-discounted} 2TBSGs. His result applies, in fact, to \emph{Parity Games} (PGs), which are an extremely restricted class of \emph{deterministic} non-discounted 2TBSG. His result implies, in particular, that the number of iterations performed by the strategy iteration algorithm may be exponential when the discount factor is considered to be part of the input. Fearnley \cite{Fearnley10}, building on Friedmann's result, showed that Howard's policy iteration may require an exponential number of iterations on non-discounted MDPs.

Ye's \cite{Ye10} results and our results, combined with the recent results of
Friedmann \cite{Friedmann09} and Fearnley \cite{Fearnley10}, supply a \emph{complete characterization} of the complexity of the policy/strategy iteration algorithm for MDPs/2TBSGs. The policy/strategy iteration algorithms are \emph{strongly polynomial} for a fixed discount factor, but \emph{exponential} for \emph{non-discounted} problems, or when the discount factor is part of the input. (In non-discounted problems the discounting criteria is replaced by \emph{limiting average} criteria. In a sense, this is equivalent to letting the discount factor tend to~$1$. See, e.g., Derman \cite{Derman72}.)

The rest of this paper is organized as follows. In Section~\ref{S-2TBSG} we define the \emph{2-player turn-based stochastic games} (2TBSG) studied in this paper. In Sections~\ref{S-basic}, \ref{S-value} and~\ref{sec-strat} we summarize known results regarding these games. For completeness, these sections contain concise, but complete, proofs of all results. (The proofs in these three sections are \emph{not} the innovative part of this paper and may be skipped at first reading.) Finally, in Section~\ref{S-strong} we obtain our innovative \emph{strongly polynomial} bound on the complexity of the celebrated strategy iteration algorithm, solving a long-standing open problem. We end in Section~\ref{S-concl} with some concluding remarks and open problems.

\section{2-player turn-based stochastic games}\label{S-2TBSG}

% \subsection{Definitions}

% A 2-player turn-based stochastic game (2TSG) is a 2-player version of the MDPs considered in the previous section. Formally, it is a tuple $G=(S_1,S_2,A,c,p,\gamma)$, where the only difference with respect to the definition of Section~\ref{S-MDP} is that the state set $S$ is divided into two sets $S_1$ and $S_2$, where $S_1$ is the set of states controlled by player~1, and $S_2$ is the set of states controlled by player~2. We also let $S=S_1\cup S_2$.

% Without loss of generality, we may assume that $S=[n]$, and that $S_1=[n_1]$ and $S_2=[n_1+1,n]$, where $n_1=|S_1|$. We also let $n_2=|S_2|=n-n_1$. A 2TSG is then compactly represented by the action matrix $Q=J-\gamma P$, the cost vector $\cc$, and the number $n_1$.

Discounted stochastic games were first studied by Shapley
\cite{Shapley53}. In his games, the players perform \emph{simultaneous}, or \emph{concurrent}, actions. We consider the subclass of \emph{turn-based} stochastic games.

We briefly review the informal definition of 2-Player Turn-Based Stochastic Games (2TBSGs), before giving a formal definition. A game is composed of states and actions.
It starts at some initial state and proceeds, in discrete steps, indefinitely.
In each time step one of the players plays an action. (The game is thus a \emph{turn-based} or \emph{perfect information} game.)
Each action has a \emph{cost} associated with it. This is the cost paid by player~1 to player~2 when this action is played. (The game is therefore a \emph{zero-sum} game.)
Each action also has a \emph{probability distribution} on states associated with it. The next state, after playing a particular action, is chosen randomly according to this probability distribution. (The game is, in general, \emph{stochastic}.)
Finally, the game is \emph{discounted}. The first player tries to minimize the expected total discounted cost, while the second player tries to maximize it.

\begin{definition}[\bf Actions] An \emph{action} $a$ over a set of states $S$ is composed of a triplet $(s(a),p(a),c(a))$, where $s(a)\in S$ is the state from which $a$ can be played, $p(a)\in \Delta(S)$ is a probability distribution over states according to which the next state is chosen when~$a$ is played, and $c(a)\in \RR$ is  the \emph{cost} of~$a$.
\end{definition}

% Note that actions may be viewed as a stochastic generalization of weighted directed edges.

\begin{definition}[\bf 2-Player Turn-Based Stochastic Games]
A \emph{2-Player Turn-Based (Discounted) Stochastic Game} (2TBSG) is a tuple $G=(S_1,S_2,A,\gamma)$, where $S_1$ and $S_2$ are the set of states controlled by players~$1$ and~$2$, respectively, and $A$ is a set of actions. We assume that $S_1\cap S_2=\emptyset$ and let $S=S_1\cup S_2$. For every $i\in S$, we let $A_i=\{a\in A\mid s(a)=i\}$ be the set of actions that can be played from~$i$. We assume that $A_i\ne \emptyset$, for every $i\in S$. We let $A^1=\cup_{i\in S_1} A_i$ and $A^2=\cup_{i\in S_2} A_i$ be the sets of all actions that can be played by players~$1$ and~$2$, respectively. Finally, $0<\gamma<1$ is a fixed \emph{discount factor}. If the infinite sequence of actions taken by the two players is $a_0,a_1,\ldots$, then the \emph{total discounted cost} of this action sequence is $\sum_{k\ge 0} \gamma^k c(a_k)$.
\end{definition}

% A \emph{2-Player Turn-Based Stochastic Game} (2TBSG) is a game played by two players. The game is composed of a finite set $S$ of \emph{states} and a finite set $A$ of \emph{actions}. The state set $S$ is partitioned $S=S_1\dotcup S_2$ into two disjoint sets $S_1$ and $S_2$, where $S_j$ is the set of states controlled by player~$j$, for $j\in\{1,2\}$. Each action $a\in A$ has a specific state $s(a)\in S$ from which it can be played. For every $i\in S$, we let $A_i=s^{-1}(i)=\{a\in A\mid s(a)=i\}$ be the set of actions that can be played from state~$i$. We assume that $A_i\ne\emptyset$, for every $i\in S$.
% When the play is at a state $i\in S_j$, where $j\in\{1,2\}$, player~$j$ chooses an action $a\in A_i$. Playing action~$a$

% A game~$G$ is played as follows.
% The play starts at some initial state $i_0\in S$ and goes on indefinitely. When the game is in state $i\in S_j$, player~$j$ chooses an action $a\in A_i$. The next state is then determined according to the probability distribution $p(a)$. The goal of player~1 is to \emph{minimize} the expected \emph{discounted} cost of all the actions taken, while the goal of player~2 is to \emph{maximize} it. If the infinite sequence of actions taken by the two players is $a_0,a_1,\ldots$, then the discounted cost of this action sequence is $\sum_{k\ge 0} \gamma^k c(a_k)$, where
% $0<\gamma<1$ is the fixed \emph{discount factor}.

If one of the players has only a single action available from each
state under her control,
the game degenerates into a 1-player game
known as a \emph{Markov Decision Process}. (This happens, in
particular, when $S_1=\emptyset$ or $S_2=\emptyset$.)

% \begin{definition}[Markov Decision Processes (1-player games)]
% A game $G=(S_1,S_2,A,s,p,c,\gamma)$ in which $|A_i|=1$ for every $i\in S_2$ is said to be a \emph{minimization} Markov Decision Process (MDP). A game $G=(S_1,S_2,A,c,p,\gamma)$ in which $|A_i|=1$ for every $i\in S_1$ is said to be a \emph{maximization} Markov Decision Process (MDP).
% \end{definition}

We next define the \emph{probability} and \emph{action} matrices of 2TBSGs. These matrices provide a compact representation of 2TBSGs that greatly simplifies their manipulation.
Throughout the paper, we use $n=|S|$ and $m=|A|$ to denote the number of states and actions, respectively, in a game.

\begin{definition}[\bf Probability and action matrices] Let $G=(S_1,S_2,A,\gamma)$ be a 2TBSG.
We assume, without loss of generality, that $S=S_1\cup S_2=[n]$ and $A=[m]$.
% (We have $S_1,S_2\subseteq S$ and $A_i\subseteq A$, for $i\in S$.)
%
We let $P\in \RR^{m\times n}$, where $P_{a,i}=p(a)_i$ is the probability of ending up in state~$i$ after taking action~$a$, for every $a\in A=[m]$ and $i\in S=[n]$, be the \emph{probability matrix} of the game, and $\cc \in \RR^m$, where $\cc_a=c(a)$ is the cost of action $a\in A=[m]$, be its \emph{cost vector}.
We also let $J\in \RR^{m\times n}$ be a matrix such that $J_{a,i}=1$ if and only if $a\in A_i$, and~$0$ otherwise.
Finally, we let $Q=J-\gamma P$ be the \emph{action matrix} of~$G$.
\end{definition}

It is interesting to note that a 2TBSG is fully specified by its action matrix $Q=J-\gamma P$, its cost vector~$\cc$,
and the partition of $S=[n]$ into $S_1$ and $S_2$.
(Action matrices may be thought of as a stochastic and discounted generalization of the \emph{incidence matrices} of directed graphs.)

\begin{definition}[\bf Strategies, strategy profiles] A (positional) \emph{strategy} $\pi_j$ for player~$j$, is a mapping $\pi_j:S_j\to A$ such that $\pi_j(i)\in A_i$, for every $i\in S_j$. We say that player~$j$ uses strategy $\pi_j$ if whenever the game is in state~$i$, player~$j$ chooses action $\pi_j(i)$. % (Here $j\in\{1,2\}$.)
A \emph{strategy profile} $\pi=(\pi_1,\pi_2)$ is simply a pair of strategies for the two players.
We let $\Pi_j=\Pi_j(G)$, for $j\in\{1,2\}$, be the set of all strategies of player~$j$, and let $\Pi=\Pi(G)=\Pi_1\times \Pi_2$ be the set of all strategy profiles in~$G$.
% Note that $\pi$ may be viewed as a strategy in the MDP obtained
\end{definition}

We note that a strategy profile $\pi=(\pi_1,\pi_2)$ may be viewed as a
mapping $\pi:S\to A$, i.e., as a strategy in a 1-player version of the game.
All strategies considered in this paper are positional.
When convenient, we also view a strategy~$\pi_j$ or a strategy profile~$\pi$
as subsets $\pi_j(S),\pi(S)\subseteq A$.
A strategy profile $\pi=(\pi_1,\pi_2)$, when viewed as a subset of~$A$, is simply the \emph{union} $\pi_1\cup \pi_2$.
We let $P_\pi \in \RR^{n\times n}$ be the matrix obtained by selecting the \emph{rows} of~$P$ whose indices belong to~$\pi$. Note that~$P_\pi$ is a (row) \emph{stochastic matrix}. Its elements are non-negative and the elements in each row sum to~$1$.
Similarly, $\cc_\pi \in \RR^{n}$ is the vector containing the costs of the actions that belong to~$\pi$. We conveniently have $J_\pi=I$ and $Q_\pi=I-\gamma P_\pi$, for every strategy profile~$\pi$.

% \begin{definition}[Induced MDPs]
% Let $G=(S_1,S_2,A,c,p,\gamma)$ be a 2TBSG and let $\pi_2$ be a strategy for player~$2$. Then, $G_{\star,\pi_2}$ is the (minimization) MDP $G=(S_1,S_2,A\cap \pi_2,c,p,\gamma)$ in which player~$2$ is forced to play according to~$\pi_2$. Similarly, if $\pi_1$ is a strategy for player~$1$, then $G_{\pi_1,\star}$ is the (maximization) MDP in which player~$1$ is forced to play according to~$\pi_1$.
% \end{definition}

\begin{definition}[\bf Value vectors] For every strategy profile $\pi=(\pi_1,\pi_2)\in \Pi$, we let $\vv_\pi=\vv_{\pi_1,\pi_2}\in \RR^n$ be a vector such that $(\vv_\pi)_i$, for every $i\in S$, is the \emph{expected} total discounted cost when the game starts at state~$i$, player~$1$ uses strategy~$\pi_1$, and player~$2$ uses strategy~$\pi_2$.
\end{definition}

Given two vectors $\uu,\vv \in \RR^n$, we say that $\uu \le \vv$ if and only if $\uu_i\le \vv_i$, for every $1\le i\le n$. We say that $\uu<\vv$ if and only if $\uu\le \vv$ and $\uu\ne \vv$.

% \begin{definition}[Optimal strategies for MDPs]
% A strategy $\pi^*$ is said to be an \emph{optimal} strategy for a minimization MDP if and only if $\vv_{\pi^*} \le \vv_{\pi}$, for every $\pi \in \Pol=\Pol_1$. Similarly, a strategy $\pi^*$ is said to be an \emph{optimal} strategy for a maximization MDP if and only if $\vv_{\pi^*} \ge \vv_{\pi}$, for every $\pi \in \Pol=\Pol_2$.
% \end{definition}

% The existence of optimal strategies for MDPs is not immediate.
% The following theorem, however, is well known.
% (Several different proofs of the theorem would appear in the sequel.)
% (We in fact, give below three different proofs of it.)

% \begin{theorem}\label{T-MDP} Every (minimization or maximization) MDP has an optimal policy~$\pi^*$.
% \end{theorem}

% Optimal policies are not necessarily unique, but the optimal value vector $\vv^*=\vv_{\pi^*}$ is, of course, unique.

\begin{definition}[\bf Optimal counter strategies] Let $G$ be a 2TBSG and
let $\pi_2\in \Pol_2(G)$ be a strategy of player~$2$. A strategy $\pi_1$ for player~$1$ is said to be an \emph{optimal counter-strategy} against~$\pi_2$, if and only if
$\vv_{\pi_1,\pi_2}\le \vv_{\pi'_1,\pi_2}$, for every $\pi'_1\in \Pol_1(G)$. Similarly,
a strategy $\pi_2$ for player~$2$ is said to be an \emph{optimal counter-strategy} against~$\pi_1$, if and only if
$\vv_{\pi_1,\pi_2}\ge \vv_{\pi_1,\pi'_2}$, for every $\pi'_2\in \Pol_2(G)$.
For every $\pi_1\in \Pol_1(G)$, we let $\tau_2(\pi_1)$ be an optimal counter strategy against~$\pi_1$, if one exists. For every $\pi_2\in \Pol_2(G)$, we let $\tau_1(\pi_2)$ be an optimal counter strategy against~$\pi_2$, if one exists.
\end{definition}

It is not immediately clear that optimal counter strategies always exist. (Note, that $\vv_{\pi_1,\pi_2}\le \vv_{\pi'_1,\pi_2}$ and $\vv_{\pi_1,\pi_2}\ge \vv_{\pi_1,\pi'_2}$ are \emph{vector} inequalities. As defined, optimal counter strategies need to be optimal for every initial state.) Furthermore, optimal counter strategies, if they exist, need not be unique.
It is well known, however, that optimal counter strategies do always
exist, as we shall also show below.

In a two-player zero-sum game, an optimal strategy is by definition
one that secures
the best possible guarantee on the expected payoff against any
opponent. As with  finite games,
pairs of optimal strategies
in a zero-sum stochastic game coincide with the Nash equilibria of
the game. This was established by Shapley \cite{Shapley53}. For
brevity, we take this characterization to be the definition
of an optimal strategy.

%\begin{definition}[Optimal counter strategies] Let $G$ be a 2TBSG and let
% let $\pi_2\in \Pol_2(G)$ be a strategy of player~$2$. A strategy $\pi_1$ for player~$1$ is said to be an \emph{optimal counter-strategy} against strategy $\pi_2$, if and only if $\pi_1$ is an optimal strategy for player~$1$ in the induced MDP $G_{\star,\pi_2}$. Similarly, a strategy $\pi_2$ of player~$2$ is said to be an \emph{optimal counter-strategy} against strategy $\pi_1$, if and only if $\pi_2$ is an optimal strategy for player~$2$ in the induced MDP $G_{\pi_1,\star}$. For every $\pi_1\in \Pol_1$, we let $\tau_2(\pi_1)\in \Pol_2$ be an optimal counter-strategy against $\pi_1$, and for every $\pi_2\in \Pol_2$, we let $\tau_1(\pi_2)\in \Pol_1$ be an optimal counter-strategy against $\pi_2$. (Note that optimal counter-strategies are not necessarily unique.)
% \end{definition}

\begin{definition}[\bf Optimal strategies] A strategy profile $\pi=(\pi_1,\pi_2)\in \Pol(G)$ is said to be \emph{optimal} if and only if $\pi_1$ is an optimal counter strategy against $\pi_2$, and $\pi_2$ is an optimal counter strategy against~$\pi_1$. In such a case we also say that $\pi_1$ is an optimal strategy for player~$1$ and that $\pi_2$ is an optimal strategy for player~$2$.
\end{definition}

Shapley \cite{Shapley53} also established the following theorem.

\begin{theorem}\label{T-opt} Every 2TBSG has an optimal strategy profile. If $\pi$ and $\pi'$ are two optimal strategy profiles then $\vv_\pi=\vv_{\pi'}$.
\end{theorem}

Theorem~\ref{T-opt} immediately implies the existence of optimal counter strategies against any strategy.
It is easy to see that $\pi_1$ is an optimal strategy for player~$1$ if and only if $\vv_{\pi_1,\tau_2(\pi_1)} \le \vv_{\pi'_1,\tau_2(\pi'_1)}$, for every $\pi'_1\in \Pol_1$. An analogous condition clearly holds for player~$2$.
The main result of this paper is a proof that a pair of optimal
strategies can be computed in \emph{strongly} polynomial time, when
the discount factor is constant. % Furthermore, this can be done using

%\begin{definition}[Optimal strategies for 2TBSG] A strategy $\pi^*_1\in \Pol_1$ is said to be an \emph{optimal strategy} for player~$1$ if and only if $\vv_{\pi^*_1,\tau_2(\pi^*_1)} \le \vv_{\pi_1,\tau_2(\pi_1)}$, for every strategy $\pi_1\in \Pol_1$. Similarly, a strategy $\pi^*_2\in \Pol_2$ is said to be an \emph{optimal strategy} for player~$2$ if and only if $\vv_{\tau_1(\pi^*_2),\pi^*_2} \ge \vv_{\tau_1(\pi_2),\pi_2}$, for every strategy $\pi_2\in \Pol_2$. A strategy profile $\pi^*=(\pi^*_1,\pi^*_2)$ is said to be optimal if $\pi^*_1$ is an optimal strategy for player~$1$ and $\pi^*_2$ is an optimal strategy for player~$2$.
%\end{definition}

% It is not difficult to check that if $\pi^*=(\pi^*_1,\pi^*_2)$ is an optimal strategy profile, then $\pi^*_1$ is an optimal counter-strategy against $\pi^*_2$, and $\pi^*_2$ is an optimal counter-strategy against $\pi^*_1$.

% The following Theorem, which is a generalization of Theorem~\ref{T-MDP}, is also well-known.

% \begin{theorem}\label{T-2TBSG} Every 2TBSG has an optimal strategy profile $\pi^*=(\pi^*_1,\pi^*_2)$.
% \end{theorem}

\section{Basic results}\label{S-basic}

For any strategy profile $\pi$, the matrix $(I-\gamma P_\pi)$ plays a prominent role in the sequel. (Recall that $P_\pi$ is the matrix obtained by selecting the rows of $P$ that correspond to actions that belong to~$\pi$.) We thus start with the following lemma whose trivial proof is omitted.

\begin{lemma}\label{L-gP-2} For any strategy profile $\pi$, the matrix $(I-\gamma P_\pi)$ is invertible and $$(I-\gamma P_\pi)^{-1}= \sum_{k\ge 0} (\gamma P_\pi)^k.$$ All entries of $(I-\gamma P_\pi)^{-1}$ are non-negative and the entries on the diagonal are strictly positive.
\end{lemma}

\begin{lemma}\label{L-value-vector-2}
For every strategy profile $\pi\in \Pi$ and every $0<\gamma<1$, we have
%the value vector~$\vv_\pi$ is %well defined and is  equal to
$$\vv_\pi = (I-\gamma P_\pi)^{-1} \cc_\pi.$$
\end{lemma}

\begin{proof} When the players use the strategy profile~$\pi$, the process becomes a Markov chain with rewards with transition matrix~$P_\pi$. In particular, for every $i,j\in [n]$ and every $k\ge 0$, $(P^k_\pi)_{i,j}$ is the probability that a game that starts at state~$i$ is in state~$j$ after exactly $k$ steps. The expected total discounted costs, starting from all states are thus
\[\displaystyle\vv_\pi \;=\; \biggl(\sum_{k\ge 0} (\gamma^k P_\pi^k)\biggr) \cc_\pi \;=\; (I-\gamma P_{\pi})^{-1} \cc_\pi.\qedhere\]
\end{proof}

\begin{definition}[\bf Modified costs]\label{D-modified-2}
The \emph{modified cost} vector $\cc^{\pi}\in \RR^m$ corresponding to a strategy profile~$\pi$ is defined to be
% defined as
% $$\s^\pi \;=\; (\cc+\gamma P \vv_\pi) - \vv^\pi \;=\; \cc - (J-\gamma P)\vv_\pi .$$
$$\cc^\pi \;=\; \cc - (J-\gamma P)\vv_\pi .$$
\end{definition}

The \emph{modified cost vector} $\cc^\pi$ is obtained from $\cc$ via a \emph{potential transformation} that uses $\vv_\pi$ as a vector of potentials. (If $h:V\to\RR$ is a function assigning potentials to the states, then the modified cost~$c_h(a)$ is defined as $c_h(a)=c(a)-h(a)+\gamma\sum_{j\in S} p(a)_j h(j)$.)

It is important to stress the difference between $\cc_\pi\in \RR^n$, the vector obtained by selecting the entries of $\cc$ corresponding to strategy profile~$\pi$, and the modified cost vector $\cc^\pi = \cc - (J-\gamma P)\vv_\pi\in \RR^m$ of Definition~\ref{D-modified-2}. (This distinction may be confusing at first, but it is extremely useful.) %, as we shall see below.)
% It is also important to note that $(\cc^\pi)_\pi=\zz$, where $\zz$ is a zero vector.

We let $\zz$ be an all zero vector. (The dimension of $\zz$ will depend on the context.)
Using Lemma~\ref{L-value-vector-2} we immediately get the following basic but important relation. % $(\cc^\pi)_\pi=\zz$, where $\zz$ is a zero vector.

\begin{lemma}\label{L-cpi} For every strategy profile $\pi$ we have $(\cc^\pi)_\pi=\zz$.
\end{lemma}

\begin{definition}[\bf Modified value vectors]
For every two strategy profiles $\pi,\pi'$, we let $\vv^{\pi}_{\pi'}$ be the value vector of $\pi'$ corresponding to the modified cost vector $\cc^\pi$.
\end{definition}

\begin{lemma}\label{L-modified-2} For every two strategy profiles $\pi',\pi$ we have $$\vv^{\pi}_{\pi'} = \vv_{\pi'}-\vv_\pi.$$
\end{lemma}

\begin{proof}By Definition~\ref{D-modified-2} and Lemma~\ref{L-value-vector-2} we have
\[\begin{array}[b]{lll}
\vv^\pi_{\pi'} &=& (I-\gamma P_{\pi'})^{-1}(\cc^\pi)_{\pi'} \\
& = & (I-\gamma P_{\pi'})^{-1} (\cc_{\pi'}-(I-\gamma P_{\pi'})\vv_\pi)\\
& = & \vv_{\pi'}-\vv_{\pi}.
\end{array}\qedhere\]
\end{proof}

Recall that $A^1=\cup_{i\in S_1} A_i$ and $A^2=\cup_{i\in S_2} A_i$.

\begin{lemma}\label{L-opt-2}{\bf (Optimality condition)} A strategy profile $\pi$ is optimal iff $(\cc^\pi)_{A^1}\ge \zz$ and $(\cc^\pi)_{A^2}\le \zz$.
\end{lemma}

\begin{proof} Suppose that $(\cc^\pi)_{A^1}\ge \zz$ and $(\cc^\pi)_{A^2}\le \zz$.  Let $\pi=(\pi_1,\pi_2)$. We prove that $\pi_1$ is an optimal counter strategy against $\pi_2$.
By Lemma~\ref{L-cpi} we have
$(\cc^\pi)_{\pi_1}=\zz$,  $(\cc^\pi)_{\pi_2}=\zz$ and hence $\vv^\pi_{\pi_1,\pi_2}=\zz$. For every $\pi'_1\in \Pol_1$, we have $(\cc^\pi)_{\pi'_1}\ge \zz$, as $\pi'_1\subseteq A^1$, and hence $(\cc^\pi)_{\pi'_1,\pi_2}\ge \zz$. Thus clearly $\vv^\pi_{\pi'_1,\pi_2}\ge \zz = \vv^\pi_{\pi_1,\pi_2}$, and $\pi_1$ is indeed an optimal counter strategy against~$\pi_2$. The proof that $\pi_2$ is an optimal counter strategy against $\pi_1$ is analogous.

Suppose now that there is an action $a\in A_{i_0}$, where $i_0\in S_1$, such that $(\cc^\pi)_a<0$. (The case in which $i_0\in S_2$ and $(\cc^\pi)_a>0$ is analogous.) Again, let $\pi=(\pi_1,\pi_2)$. Let $\pi'_1\in \Pol_1$ be a policy such that $\pi'_1(i)=\pi_1(i)$, if $i\ne i_0$, and $\pi'_1(i_0)=a$. We then have $(\cc^\pi)_{\pi'_1}<\zz$ and $(\cc^\pi)_{\pi_2}=\zz$. Thus $\vv^\pi_{\pi'_1,\pi_2}<\zz$. (The strict inequality follows from Lemma~\ref{L-gP-2}. All entries of $(I-\gamma P_{\pi'_1,\pi_2})^{-1}$ are non-negative, and the entries on the diagonal are strictly positive.) Thus $\pi_1$ is \emph{not} an optimal counter strategy against~$\pi_2$.
\end{proof}

In the second part of the proof above, $\pi'_1$ is obtained from
$\pi_1$ by a \emph{profitable switch}. Profitable switches
are closely related to the pivoting steps performed by the simplex algorithm. They also
lie at the core of the strategy iteration algorithm whose analysis is the main focus of this paper.

\begin{definition}[\bf Flux vectors]\label{D-flux-2} For every strategy profile~$\pi$, let $\xx_\pi\in \RR^{1\times n}$ be a row vector such that~$(\xx_\pi)_i$, for every $i\in S$, is the sum of the discounted costs, over all states, when the cost of action~$\pi(i)$ is~$1$, while the cost of all other actions is~$0$, and when the players use strategy profile~$\pi$.
% We also let $\xx^\pi \in \RR^{1\times m}$ be a row vector such that $(\xx^\pi)_\pi=\xx_\pi$ and such that all its other entries are~$0$.
\end{definition}

%Note that if $a\not\in \pi$, then $(\xx^\pi)_a=0$. Thus, the elements of $\xx_\pi$ are the only non-zero elements in $\xx^\pi$.

We let $\ee=(1,1,\ldots,1)^T\in \RR^n$ be an all one vector.
Using Lemma~\ref{L-value-vector-2}, we easily get

\begin{lemma}\label{L-flux-2} For every strategy profile~$\pi$, we have
$$\xx_\pi \;=\; \ee^T (I-\gamma P_\pi)^{-1}.$$
\end{lemma}

It is in fact possible to view Lemma~\ref{L-flux-2} as the definition of $\xx_\pi$. The meaning of the flux vectors given in Definition~\ref{D-flux-2} is not used in the sequel. (The flux vectors are intimately related to the \emph{dual} linear program formulation of MDPs.)

\begin{lemma}\label{L-x-2} For every strategy profile~$\pi$, we have
$$\xx_\pi \ee = \frac{n}{1-\gamma}.$$
% $\ee_m^T \xx^\pi = \ee_n^T \xx_\pi = \frac{n}{1-\gamma}$.
\end{lemma}

\begin{proof}
By Lemma~\ref{L-flux-2}, Lemma~\ref{L-gP-2}, and the fact that $\ee^T(P_\pi)^k \ee=n$, for every $k\ge 0$, we have:
\[\xx_\pi \ee = \ee^T (I-\gamma P_\pi)^{-1} \ee = \sum_{k\ge 0} \ee^T(\gamma P_\pi)^k \ee = n \sum_{k\ge 0} \gamma^k = \frac{n}{1-\gamma}.\qedhere\]
\end{proof}

\begin{lemma}\label{L-cx-2}  For every strategy profile~$\pi$, we have
$$\ee^T \vv_\pi = \xx_\pi \cc_\pi.$$
\end{lemma}

\begin{proof} By Lemma~\ref{L-value-vector-2} and then Lemma~\ref{L-flux-2}, we get
$\; \ee^T \vv_\pi \;=\; \ee^T (I-\gamma P_\pi)^{-1} \cc_\pi \;=\; \xx_\pi \cc_\pi.$
\end{proof}

\begin{lemma}\label{L-xs-2}  For every strategy profile~$\pi$, we have
$$\ee^T(\vv_{\pi'}-\vv_{\pi}) = \xx_{\pi'} (\cc^{\pi})_{\pi'}.$$
\end{lemma}

\begin{proof}
By Lemma~\ref{L-modified-2} and then Lemma~\ref{L-cx-2}, we have
% \[\ee^T(\vv_{\pi'}-\vv_\pi) \;=\; \ee^T \vv^\pi_{\pi'} \;=\; \xx_{\pi'}(\cc^\pi)_{\pi'}.\qedhere\]
$\;\ee^T(\vv_{\pi'}-\vv_\pi) \;=\; \ee^T \vv^\pi_{\pi'} \;=\; \xx_{\pi'}(\cc^\pi)_{\pi'}.$
\end{proof}

\section{Value iteration}\label{S-value}

If $\xx\in \RR^m$ and $B\subseteq [m]$, we let $\min_B \xx = \min_{j\in B} \xx_j$, and similarly $\max_B \xx= \max_{j\in B} \xx_j$.
We also let $\argmin_B \xx = \argmin_{j\in B} \xx_j$ and
$\argmax_B \xx = \argmax_{j\in B} \xx_j$.

\begin{definition}[\bf Value iteration operator]
The \emph{value iteration} operator $\TT:\RR^n \to \RR^n$ is defined as follows:
% $$(\TT \vv)_i \;=\; \min\; \cc_{A_i}+\gamma P_{A_i} \vv \quad,\quad i\in S.$$
$$\textstyle (\TT \vv)_i \;=\;
\begin{cases}
\min_{A_i} \cc+\gamma P \vv\;, & \text{if $i\in S_1$,} \\
\max_{A_i} \cc+\gamma P \vv\;, & \text{if $i\in S_2$.} \\
\end{cases}
$$
%
% If $\pi\in\Pol$ is a strategy profile, we let $\TT_\pi:\RR^n\to \RR^n$ be the following operator:
% $$\TT_\pi\vv \;=\; \cc_{\pi} + \gamma P_{\pi}\vv.$$
\end{definition}

% $$(\TT_\pi)_i \;=\; \cc_{\pi(i)} + \gamma P_{\pi(i)} \vv\quad,\quad i\in S.$$

% Note that $\TT_\pi$ is simply the value iteration operator in the degenerate game in which only the actions of~$\pi$ are available.
% In particular, Lemma~\ref{L-contract-2} applies also to~$\TT_\pi$.

The operator $\TT$ is a contraction with Lipschitz constant $\gamma$.
\begin{lemma}\label{L-contract-2} For every $\uu,\vv\in \RR^n$ we have $\|\TT\uu-\TT\vv\|_\infty \le \gamma\, \|\uu-\vv\|_\infty$.
\end{lemma}

\begin{proof} Assume that $i\in S_1$ and that $(\TT\uu)_i\ge (\TT\vv)_i$. (The other cases are analogous.)
Let $a=\argmin_{A_i} c + \gamma P\uu$ and $b = \argmin_{A_i} c+\gamma P\vv$. Then,
$$\begin{array}{lll}
(\TT\uu-\TT\vv)_i & = & (\cc_a+\gamma P_{a}\uu) - (\cc_b + \gamma P_b\vv) \\
                  & \le & (\cc_b+\gamma P_{b}\uu) - (\cc_b + \gamma P_b\vv) \\
                  & = & \gamma P_b (\uu-\vv) \\
                  & \le & \gamma\, \|\uu-\vv\|_\infty.
\end{array}$$
The last inequality follows from the fact that the elements in $P_b$ are non-negative and sum-up to~$1$.
\end{proof}

Banach fixed point theorem now implies that the value iteration operator $\TT$ has a unique fixed point.

\begin{corollary}\label{L-fixed-2}
There is a unique vector $\vv^*\in \RR^n$ such that $\TT\vv^*=\vv^*$.
\end{corollary}

We next define the \emph{strategy extraction} operators that play an important role in this section, and the central role in the next section.
% } that play the central role in the next section, as they are also used here.

\begin{definition}[\bf Strategy extraction operators]
The \emph{strategy extraction} operators $\PP_1:\RR^n \to \Pol_1$
and $\PP_2:\RR^n \to \Pol_2$ and $\PP:\RR^n \to \Pol$ are
defined as follows:
\belowdisplayskip = 0pt
$$\begin{array}{ccc}
(\PP_1\vv)(i) &=& \argmin_{A_i} \cc+\gamma P\vv \quad,\quad i\in S_1, \\[2pt]
 (\PP_2\vv)(i) &=& \argmax_{A_i} \cc+\gamma P\vv \quad,\quad i\in S_2.
\end{array}$$
and
 $$\PP\vv \;=\; (\PP_1\vv,\PP_2\vv).$$
% $$(\PP\vv)_i \;=\; \begin{cases}
% (\PP\vv)_i\;, & \text{if $i\in S_1$,}\\
% (\PP\vv)_i\;, & \text{if $i\in S_1$,}\\
%
% $$\textstyle (\PP \vv)_i \;=\;
% \begin{cases}
% \argmin_{A_i} \cc+\gamma P \vv\;, & \text{if $i\in S_1$,} \\
% \argmax_{A_i} \cc+\gamma P \vv\;, & \text{if $i\in S_2$.} \\
% \end{cases}
% $$
\end{definition}

The following relation between the value iteration and strategy extraction operator is immediate.

\begin{lemma}\label{L-VP} For every $\vv\in\RR^n$ we have $\TT\vv = (\cc + \gamma P \vv)_\pi$, where $\pi=\PP\vv$.
\end{lemma}

The following simple lemma provides an interesting relation between the strategy extraction operator and modified cost vectors.

\begin{lemma}\label{L-pp} For every strategy profile $\pi$ we have
\[\begin{array}{c}
(\PP_1\vv_\pi)(i) \;=\; \textstyle\argmin_{A_i} \cc^\pi\quad,\quad i\in S_1,\\
(\PP_2\vv_\pi)(i) \;=\; \textstyle\argmax_{A_i} \cc^\pi\quad,\quad i\in S_2
.
\end{array}\]
\end{lemma}

\begin{proof} Let $\vv=\vv_\pi$. If $a\in A_i$ then,
$$(\cc^\pi)_a \;=\; \cc_a - (\vv_i - \gamma P_a\vv) \;=\; (\cc+\gamma P\vv)_a-\vv_i.$$
Thus, if $a,a'\in A_i$, then $(\cc+\gamma P\vv)_a\le (\cc+\gamma P\vv)_{a'}$ if and only if
$(\cc^\pi)_a \le (\cc^\pi)_{a'}$.
\end{proof}

The following lemma supplies a simple proof of Theorem~\ref{T-opt}.
(This is, in fact, the original proof given by Shapley \cite{Shapley53}.)

\begin{lemma} Let $\vv^*\in \RR^n$ be the unique fixed point of $\TT$ and let $\pi=\PP\vv^*$. Then, $\pi$ is an optimal strategy profile.
\end{lemma}

\begin{proof} By Lemma~\ref{L-VP}, we get that $\vv^* = \TT\vv^* = \cc_\pi + \gamma P_\pi\vv^*$. By Lemma~\ref{L-value-vector-2} we get $\vv_\pi = \vv^*$.
We next show that $\pi$ satisfies the optimality condition of Lemma~\ref{L-opt-2}, and hence is an optimal strategy profile. Suppose that $i\in S_1$ and that $a\in A_i$.
By Lemma~\ref{L-pp}, we have $\pi(i)=(\PP_1\vv^*)(i) = \argmin_{A_i} \cc^\pi$. As $(\cc^\pi)_{\pi(i)}=0$, we get that $(c^\pi)_a\ge 0$. Similarly, if $i\in S_2$ and $a\in A_i$, we get that $(c^\pi)_a\le 0$.
\end{proof}

The \emph{value iteration} algorithm, given on the left-hand side of Figure~\ref{F-algorithms-2}, repeatedly applies the value iteration operator $\TT$ to an initial vector $\uu^0\in \RR^n$, generating a sequence of vectors $(\uu^k)_{k=0}^N$, where $\uu^{k+1}= \TT\uu^k$, until the difference between two successive vectors is small enough, i.e., $\|\uu^{k-1} - \uu^k\|_{\infty} < \epsilon$.
% Note that by setting $\epsilon=0$, we get an infinite sequence of value vectors which, as we show below, converges to the optimal value vector~$\vv^*$.

\begin{lemma}\label{L-VI-2} Let $(\uu^k)_{k=0}^N$ be the sequence of value vectors generated by a call $\ValueIteration(\uu^0,\epsilon)$, for some $\epsilon>0$. Let $\vv^*$ be the optimal value vector.
 Then, for every $0\le k\le N$ we have
$$\|\uu^k-\vv^*\|_\infty \;\le\; \gamma^k\, \| \uu^0-\vv^*\|_\infty.$$
\end{lemma}

\begin{proof} By Lemma~\ref{L-contract-2} and the fact that $\TT\vv^*=\vv^*$, we have
$$\|\uu^k-\vv^*\|_\infty \;=\; \|\TT\uu^{k-1}-\TT\vv^*\|_\infty \;\le\; \gamma\, \| \uu^{k-1}-\vv^*\|_\infty.$$
The claim follows easily by induction.
\end{proof}

It follows immediately from Lemma~\ref{L-VI-2}, that for any $\uu\in \RR^n$, the infinite sequence of vectors generated by the call $\ValueIteration(\uu^0,0)$ converges to the optimal value vector $\vv^*$. Also, for every $\epsilon>0$, the call $\ValueIteration(\uu^0,\epsilon)$ eventually terminates.

\begin{figure}[t]
\begin{center}
\parbox[t]{2.8in}{
\begin{function}[H]
\dontprintsemicolon

\SetKwRepeat{Repeat}{repeat}{until}

\BlankLine
$k \gets 0$\;
\Repeat{$\|\uu^{k-1} - \uu^k\|_{\infty} < \epsilon$}{
  \BlankLine
  $\uu^{k+1} \gets \TT \uu^k$\;
  $k \gets k+1$
  \BlankLine
}
% $\pi \gets \PP \uu^k$
\BlankLine
\Return $\uu^k$
\caption{\ValueIteration($\uu^0,\epsilon$)}
\end{function}
}
\hspace*{20pt}
\parbox[t]{3.2in}{
\begin{function}[H]
\dontprintsemicolon

\SetKwRepeat{Repeat}{repeat}{until}

\BlankLine
$k \gets 0$\;
\Repeat{$\sigma^{k-1} = \sigma^k$}{
% \Repeat{$\vv^{k-1} = \vv^k$}{
  \BlankLine
  $\tau^k = \tau_2(\sigma^k)$ \;
  $\vv^k \gets \vv_{\sigma^k,\tau^k}$\;
  $\sigma^{k+1} \gets \PP_1 \vv^k$ (if possible $\sigma^{k+1} \gets \sigma^k$) \;
  $k \gets k+1$
  \BlankLine
}
\BlankLine
\Return $\sigma^k$
\caption{\StrategyIteration($\sigma^0$)}
\end{function}
}
\end{center}
\vspace*{-9pt}
\caption{The $\ValueIteration$ and $\StrategyIteration$ algorithms.}
\label{F-algorithms-2}
\end{figure}

\section{Strategy iteration}
\label{sec-strat}

The \emph{strategy iteration} algorithm is given in the right-hand side
of Figure~\ref{F-algorithms-2}. It was first described for the MDP case by
Howard \cite{Howard60} and is called \emph{policy iteration} or \emph{Howard's
algorithm} in that context. It was described for 2-player
stochastic games by Rao {\em et al.} \cite{RCN}. (Their algorithm
% \footnote{Rao {\em et al.}
actually works on more general imperfect information games
for which it is a non-terminating approximation algorithm.)

The strategy iteration algorithm receives an initial strategy~$\sigma^0$ of player~$1$, and generates a sequence $\pi^k=(\sigma^k,\tau^k)$ of strategy profiles of the two players, ending with an optimal strategy profile.
%
% (To highlight the different treatment of strategies of player~$1$ and player~$2$, we use $\sigma^k$, instead of $\pi^k_1$, to denote strategies of player~$1$ generated by the algorithm, and use $\tau^k$ to denote strategies of player~$2$ generated by the algorithm.)
%
Each iteration of the algorithm receives a strategy $\sigma^k$ and produces an \emph{improved} strategy $\sigma^{k+1}$ as follows. The algorithm first computes an optimal counter-strategy $\tau^k=\tau_2(\sigma^k)$ for player~$2$ against~$\sigma^k$.
(We assume here that this can be done in strongly polynomial time.
One way of doing it is to apply the strategy iteration algorithm on a restricted game in which $\sigma^k$ is the only strategy available to player~$1$.)
Next, it \emph{evaluates} the strategy profile $\pi^k=(\sigma^k,\tau^k)$, by solving a system of linear equations, and obtains its value vector $\vv^k=\vv_{\pi^k}$. It then lets $\sigma^{k+1}=\PP_1 \vv_{\pi^k}$.
Ties are broken, if possible, in favor of actions that are in $\sigma^k$. (This is important, as termination is not guaranteed without this provision.)
The algorithm terminates when two consecutive strategies $\sigma^k$ and $\sigma^{k+1}$ are identical.

The step $\sigma^{k+1}=\PP_1 \vv_{\pi^k}$ is the main step of the strategy iteration algorithm. As we shall (implicitly) see below, $\sigma^{k+1}$ is obtained from~$\sigma^k$ by performing a collection of improving switches.

To prove the correctness of the \StrategyIteration\ algorithm we use the following lemma. (Note that $\pi^1$ in the lemma is obtained from $\pi^0$ using one iteration of the \StrategyIteration\ algorithm.)

\begin{lemma}\label{L-SI-1} Let $\sigma^0\in \Pol_1$, $\pi^0=(\sigma^0,\tau_2(\sigma^0))$ and $\sigma^1=\PP_1 \vv_{\pi^0}$, $\pi^1=(\sigma^1,\tau_2(\sigma^1))$. Then $\vv_{\pi^0} \ge \vv_{\pi^1}$.
\end{lemma}

\vspace*{-10pt}
\begin{proof}
We show that $\vv^{\pi^0}_{\pi^0}=\zz\ge \vv^{\pi^0}_{\pi^1}$, which by Lemma~\ref{L-modified-2} implies that $\vv_{\pi^0}\ge \vv_{\pi^1}$. To show that $\vv^{\pi^0}_{\pi^1}\le \zz$, we show that $(\cc^{\pi^0})_{\pi^1}\le \zz$. The fact that $(\cc^{\pi^0})_{\sigma^1}\le \zz$ follows from the fact that for every $i\in S_1$ we have $\sigma^1(i)=\argmin_{A_i} \cc^{\pi^0}$ and hence
$(\cc^{\pi^0})_{\sigma^1(i)}\le (\cc^{\pi^0})_{\sigma^0(i)}=0$.
The fact that $(\cc^{\pi^0})_{\tau^1}\le \zz$ follows from fact that $\tau^0$ is an optimal counter strategy against $\sigma^0$, so in fact $(\cc^{\pi^0})_{A^2}\le \zz$.
\end{proof}

% It is important to note that $\tau^k=\tau_2(\sigma^k)$ while $\sigma^{k+1}=\PP_1 \vv_{\sigma^k,\tau^k}$.

% It is well know, and it will also follow from the discussion below, that $\StrategyIteration(\sigma^0)$ terminates after a finite number of iterations, for any initial policy $\sigma^0$. As our main result, we obtain the first strongly polynomial bound on the maximal number of iterations of the \StrategyIteration\ algorithm for solving 2TBSGs.

% We note that the variant of the \StrategyIteration\ algorithm of Figure~\ref{F-algorithms-2} in which the termination condition $\vv^{k-1}=\vv^k$ is replaced by $\pi^{k-1}=\pi^k$ may fail to terminate if the optimal policy is not unique. This can be fixed by stipulating that ties that arise in the computation of $\PP \vv^k$ are resolved, if possible, in favor of actions contained in $\pi^{k}$.

\begin{lemma} For every initial strategy $\sigma^0$, $\StrategyIteration(\sigma^0)$ terminates after a finite number of iterations. If $(\vv^k)_{k=0}^N$ is the sequence of value vectors generated by the call, then, $\vv^{k-1}> \vv^k\ge \vv^{*}$, for every $1\le k<N$. Furthermore, $\vv^{N-1}=\vv^N=\vv^*$ and $\pi^{N-1}=\pi^N$ is an optimal strategy profile.
\end{lemma}

\vspace*{-10pt}
\begin{proof} The claim $\vv^{k-1}\ge \vv^k$, for every $1\le k\le N$ follows easily from Lemma~\ref{L-SI-1} by induction.
Next, we note that if $\vv^{k-1}=\vv^k$, for some $k$, then by the reasoning used in the proof of Lemma~\ref{L-SI-1}, we must have $(\cc^{\pi^{k-1}})_{A^1}\ge \zz$ and $(\cc^{\pi^{k-1}})_{A^2}\le \zz$. By the optimality condition, we get that $\pi^{k-1}$ is an optimal strategy profile. By the tie breaking mechanism used, we also get that $\pi^k=\pi^{k-1}$.
Finally, the fact that $\vv^{k-1}>\vv^k$, for every $1\le k<N$, implies that strategy profiles encountered cannot repeat themselves. As there is only a finite number of such profiles, the sequence of strategy profiles generated must be finite.
\end{proof}

We next relate the sequences of value vectors obtained by running
$\StrategyIteration(\sigma^0)$ and
$\ValueIteration(\vv_{\pi^0},\epsilon)$, where
$\pi^0=(\sigma^0,\tau_2(\sigma^0))$. The following lemmas
for the case of MDPs are well-known and appear, e.g., in
Meister and Holzbaur \cite{Meister86}. The proofs for the 2-player case
are essentially identical. (They may be folklore.)
% may be folklore (we do not know a reference--but is essentially identical.

\begin{lemma}\label{L-SI-2} Let $\sigma^0\in \Pol_1$, $\pi^0=(\sigma^0,\tau_2(\sigma^0))$, and $\sigma^1=\PP_1 \vv_{\pi^0}$, $\pi^1=(\sigma^1,\tau_2(\sigma^1))$. Then  $\TT\vv_{\pi^0} \ge \vv_{\pi^1}$.
\end{lemma}

\begin{proof}
\abovedisplayskip = 7pt
\belowdisplayskip = 7pt
Let $i\in S_1$. As $\sigma^1(i)=\argmin_{A_i}\cc+\gamma P \vv_{\pi^0}$, $\vv_{\pi^0} \ge \vv_{\pi^1}$, and $\cc_{\pi^1}+\gamma P_{\pi^1} \vv_{\pi^1}= \vv_{\pi^1}$, we have
$$\textstyle (\TT\vv_{\pi^0})_i \;=\; \min_{A_i} \cc+\gamma P \vv_{\pi^0} \;=\;
(\cc+\gamma P \vv_{\pi^0})_{\sigma^1(i)} \;\ge\; (\cc+\gamma P \vv_{\pi^1})_{\sigma^1(i)} \;=\; (\vv_{\pi^1})_i.$$
Similarly, if $i\in S_2$, then
\[\textstyle (\TT\vv_{\pi^0})_i \;=\; \max_{A_i} \cc+\gamma P \vv_{\pi^0} \;\ge\;
(\cc+\gamma P \vv_{\pi^0})_{\tau^1(i)} \;\ge\; (\cc+\gamma P \vv_{\pi^1})_{\tau^1(i)} \;=\; (\vv_{\pi^1})_i.\qedhere\]
\end{proof}

Using Lemma~\ref{L-SI-2}, we immediately get:

\begin{lemma}\label{L-PI-seq}
Let $(\vv^k)_{k=0}^N$ be the value vectors generated by %the call
$\StrategyIteration(\sigma^0)$, and let $(\uu^k)_{k=0}^\infty$ be the value vectors generated by % the call
$\ValueIteration(\vv_{\pi^0},0)$, where $\pi^0=(\sigma^0,\tau_2(\sigma^0))$. Then, $\vv^k\le \uu^k$, for every $0\le k\le N$.
\end{lemma}

\begin{proof} We prove the lemma by induction. We have $\vv^0=\uu^0$. Suppose now that $\vv^k\le \uu^k$. Then, by Lemma~\ref{L-SI-2} and the monotonicity of the value iteration operator, we have:
\[ \vv^{k+1} \;\le\; \TT\vv^k \;\le\; \TT\uu^k \;=\; \uu^{k+1}.\qedhere\]
\end{proof}
% \begin{lemma}\label{L-PIn} Let $(\vv^k)_{k=0}^N$ be the sequence of value vectors generated by a % call $\PolicyIteration(\pi^0)$, for some initial policy $\pi^0\in \Pi$. Let $\vv^*$ be the
% optimal value vector.
%  Then, for every $0\le k\le N$ we have
% $$\|\vv^k-\vv^*\|_\infty \;\le\; \gamma^k\, \| \vv^0-\vv^*\|_\infty.$$
% \end{lemma}

Combining Lemmas~\ref{L-VI-2} and~\ref{L-PI-seq}, we get

\begin{lemma}\label{L-PI-2} Let $(\vv^k)_{k=0}^N$ be the sequence of value vectors generated by $\StrategyIteration(\sigma^0)$, for some $\sigma^0\in \Pol_1$. Let $\vv^*$ be the optimal value vector. Then, for every $0\le k\le N$ we have
% $$\ee^T(\vv^k-\vv^*) \;\le\; n\gamma^k\, \ee^T (\vv^0-\vv^*).$$
$$\|\vv^k-\vv^*\|_\infty \;\le\; \gamma^k\, \|\vv^0-\vv^*\|_\infty.$$
\end{lemma}

\section{Strongly polynomial bound}\label{S-strong}

In this section, the main section of the paper, we present our strongly polynomial bound on the number of iterations performed by the strategy iteration algorithm. We begin with some technical lemmas.

\begin{lemma}\label{L-G1} Let $\pi',\pi$ be two strategy profiles such that $\vv_{\pi'}\ge \vv_\pi$ and let $a=\pi'(i)$ where $i\in S$.
Then, $$(\vv_{\pi'}-\vv_\pi)_i \;\ge\; (\cc^\pi)_a.$$
\end{lemma}

\begin{proof}
% $$\begin{array}{lll}
% (\vv_{\pi'})_i - (\vv_\pi)_i & = & (\cc_a+\gamma P_a \vv_{\pi'}) - (\vv_\pi)_i \\
% & \ge & (\cc_a+\gamma P_a \vv_{\pi}) - (\vv_\pi)_i \\
% & = & (\cc^\pi)_a.
% \end{array}$$
$\quad\quad\quad (\vv_{\pi'})_i - (\vv_\pi)_i \; = \; (\cc+\gamma P \vv_{\pi'})_a - (\vv_\pi)_i
\;\ge\; (\cc+\gamma P \vv_{\pi})_a - (\vv_\pi)_i \;=\; (\cc^\pi)_a.$
\end{proof}

\begin{lemma}\label{L-G2} Let $\pi'',\pi$ be two strategy profiles such that $\vv_{\pi''}\ge \vv_{\pi}$ and let $a=\argmax_{\pi''} \cc^\pi$.
Then,%  $\ee^T(\vv_{\pi''}-\vv_\pi) \le  \frac{n}{1-\gamma} (\cc^\pi)_a $.
$$\|\vv_{\pi''}-\vv_\pi\|_1 \;\le\;  \frac{n}{1-\gamma} (\cc^\pi)_a .$$
\end{lemma}

\begin{proof} As $\vv_{\pi''}\ge \vv_{\pi}$, we get using
Lemma~\ref{L-xs-2} and then Lemma~\ref{L-x-2} that
\[\|\vv_{\pi''}-\vv_\pi\|_1 \;=\; \ee^T(\vv_{\pi''}-\vv_\pi) \;=\; \xx_{\pi''}(\cc^\pi)_{\pi''} \;\le\; \xx_{\pi''} \ee\, (\cc^\pi)_a \;=\; \frac{n}{1-\gamma} (\cc^\pi)_a .\qedhere\]
\end{proof}

\begin{lemma}\label{L-G3} Let $\pi'',\pi',\pi$ be three strategy profiles such that $\vv_{\pi''}\ge \vv_{\pi'}\ge \vv_{\pi}$. Let $a=\argmax_{\pi''} \cc^\pi$ and suppose that % $(\cc^\pi)_a>0$ and
$a\in \pi'$. Then,
% $\ee^T(\vv_{\pi'}-\vv_{\pi}) \ge \frac{1-\gamma}{n}\, \ee^T(\vv_{\pi''}-\vv_{\pi})$.
$$\|\vv_{\pi'}-\vv_{\pi}\|_1 \;\ge\; \frac{1-\gamma}{n}\, \|\vv_{\pi''}-\vv_{\pi}\|_1.$$
\end{lemma}

\begin{proof}
\abovedisplayskip = 5pt
Let $i\in S$ be the state for which $\pi''(i)=\pi'(i)=a$.
% (We know that $a\in \pi'$, so there must be such a state.)
By Lemma~\ref{L-G1} and Lemma~\ref{L-G2} we get
\[\|\vv_{\pi'}-\vv_\pi\|_1 \;\ge\; (\vv_{\pi'}-\vv_\pi)_i \;\ge\; (\cc^\pi)_{a} \;\ge\; \frac{1-\gamma}{n}\,  \|\vv_{\pi''}-\vv_\pi\|_1.\qedhere\]
\end{proof}

\begin{lemma}\label{L-main-2} Let $(\sigma^k)_{k=0}^N$ be the sequence of player~$1$ strategies generated by the \StrategyIteration\ algorithm, starting from some initial strategy $\sigma^0$. % Let $L=\bigl(\log\frac{1-\gamma}{n^2}\bigr)/(\log \gamma)$.
% Let $L=\log_{\,\gamma}\frac{1-\gamma}{n^2}$.
Let $L=\log_{\,1/\gamma}\frac{n^2}{1-\gamma}$.
Then, every strategy $\sigma^k$ contains an action that does not appear in any strategy $\sigma^{\ell}$, where $k+L<\ell\le N$.
\end{lemma}

\begin{proof}
\abovedisplayskip = 5pt
\belowdisplayskip = 5pt
Let $(\pi^k)_{k=0}^N$, where $\pi^k=(\sigma^k,\tau^k)$, be the sequence of strategy profiles generated by the strategy iteration algorithm. By the correctness of the strategy iteration algorithm, $\pi^*=\pi^N$ is composed of optimal strategies for the two players.
% Let $\pi^*=\pi^N$ be a the optimal policy.
Let $a=\argmax_{\pi^k} \cc^{\pi^*}$.
% (The action achieving this maximum is not necessarily unique. We break ties arbitrarily.)
By Lemma~\ref{L-opt-2}, we have $(\cc^{\pi^*})_a\ge 0$ for every $a\in A^1$, and $(\cc^{\pi^*})_a\le 0$ for every $a\in A^2$. We may assume, therefore, that $a\in A^1$, i.e., that~$a$ is an action controlled by player~$1$.
%
% We claim that $a$ is an action controlled by player~$1$.
%
Suppose, for the sake of contradiction, that $a\in \pi^\ell$, for some $k+L<\ell\le N$. Using Lemma~\ref{L-G3}, with $\pi''=\pi^k$, $\pi'=\pi^\ell$ and $\pi=\pi^*$, we get that
\[\|\vv_{\pi^\ell}-\vv_{\pi^*}\|_1 \;\ge\; \frac{1-\gamma}{n}\, \| \vv_{\pi^k}-\vv_{\pi^*}\|_1.\]
On the other hand, using Lemma~\ref{L-PI-2}, we get that
\[\|\vv_{\pi^\ell}-\vv_{\pi^*}\|_\infty \;\le\; \gamma^{\ell-k} \| \vv_{\pi^k} - \vv_{\pi^*} \|_\infty.\]

\vspace*{-12pt}
Thus,
\[\| \vv_{\pi^\ell}-\vv_{\pi^*}\|_1 \;\le\; n\,\|\vv_{\pi^\ell}-\vv_{\pi^*}\|_\infty \;\le\;
n\gamma^{\ell-k}\, \|\vv_{\pi^k}-\vv_{\pi^*}\|_\infty \;\le\;
n\gamma^{\ell-k}\, \|\vv_{\pi^k}-\vv_{\pi^*}\|_1.\]
It follows that $n\gamma^{\ell-k}\ge \frac{1-\gamma}{n}$ and hence
\[\gamma^L \;>\; \gamma^{\ell-k} \;\ge\; \frac{1-\gamma}{n^2},\]
a contradiction.
\end{proof}

\begin{theorem} The \StrategyIteration\ algorithm, starting from any initial strategy, terminates with an optimal strategy after at most $(m+1)(1+\log_{\,1/\gamma}\frac{n^2}{1-\gamma}) = O(\frac{m}{1-\gamma}\log\frac{n}{1-\gamma})$ iterations.
\end{theorem}

\begin{proof} Let $\bar{L}=\lfloor 1+\log_{\,1/\gamma}\frac{n^2}{1-\gamma}\rfloor$. Consider strategies $\sigma^0,\sigma^{\bar{L}},\sigma^{2\bar{L}},\ldots$. By Lemma~\ref{L-main-2}, every strategy in this subsequence contains a new action that would never be used again. As there are only $m$ actions, the total number of strategies in the sequence is at most $(m+1)\bar{L}=(m+1)(1+\log_{\,1/\gamma}\frac{n^2}{1-\gamma})$.
Finally, note that $\log_{\,1/\gamma} x = \frac{\log x}{\log\,1/\gamma} \le \frac{x}{1-\gamma}$.
\end{proof}

% \vspace*{-15pt}
\section{Concluding remarks}\label{S-concl}

We have shown that the strategy iteration algorithm is \emph{strongly polynomial} for 2TBSGs with a
\emph{fixed} discount factor. Friedmann \cite{Friedmann09}, on the other hand, has recently shown that the strategy iteration algorithm is \emph{exponential} for non-discounted 2TBSG, or when the discount factor is part of the input.

The existence of polynomial time algorithms for 2TBSGs when the discount factor is part of the input, or for the non-discounted case, remains an intriguing and a challenging open problem, with many possible consequences for complexity theory and automatic verification. As shown by Andersson and Miltersen \cite{AnMi09}, this is equivalent to finding a polynomial time algorithm for Condon's \cite{Condon92} \emph{Simple Stochastic Games} (SSGs).
Such an algorithm will immediately provide polynomial time algorithms for \emph{Mean Payoff Games} (MPGs) (see \cite{EhMy79},\cite{GuKaKh88},\cite{ZwPa96}) and \emph{Parity Games} (PGs) (see, e.g., \cite{EmJu91}, \cite{VoJu00}, \cite{JuPaZw08}).

We believe that our results give some hope of obtaining a polynomial time algorithm for
% 2TBSGs, when the discount factor is part of the input.
this problem.
In an earlier work, Ye \cite{Ye05} gave a polynomial time
algorithm for the analogous MDP problem. His algorithm uses interior point
methods and its analysis relies again on the LP formulation of the MDP problem.
Given the ``deLPfication'' of Ye's \cite{Ye10} analysis of the policy iteration algorithm
carried out here, one could speculate that looking at interior point methods for the two-player case,
with Ye's \cite{Ye05} algorithm for MDPs as a starting point, would be a
fertile approach.

% Non-discounted case?

% \nocite{LiDeKa95}
% \nocite{MaSi99}
% \nocite{MeCo94}
% \nocite{PaTs87}
% \nocite{Tseng90}
% \nocite{BjVo05,BjVo07}
% \nocite{EhMy79}
% \nocite{GuKaKh88}
% \nocite{Halman07}
% \nocite{Ludwig95}
% \nocite{ZwPa96}

% \nocite{Bellman57}
% \nocite{Bertsekas01}
% \nocite{Derman72}
% \nocite{Howard60}

% \nocite{dEpenoux63}
% \nocite{Fearnley10}
% \nocite{Ye05}
% \nocite{Ye10}

% \nocite{AnMi09}

% \nocite{Condon92,Condon93}

% \nocite{Friedmann09}
% \nocite{HoKa66}

% \nocite{Shapley53}

% \bigskip
% {\bf Note:} Many of the references below are not cited yet.

\bibliographystyle{plain}
% \bibliography{ref}
\bibliography{book,mdp,inf-games,graph-short,linear-prog,matrix,peter}

\end{document}